\documentclass[a4paper, amsfonts, amssymb, amsmath, reprint, showkeys, nofootinbib, twoside,notitlepage,twocolumn]{revtex4-1}

\def\prd{{Phys.~Rev.~D}}        
\usepackage{amsmath,amstext}
\usepackage[T1]{fontenc}
\usepackage{amssymb}
\usepackage{graphicx}
\usepackage{ae,aecompl}

\DeclareFontFamily{OT1}{pzc}{}
\DeclareFontShape{OT1}{pzc}{m}{it}{<-> s * [1.10] pzcmi7t}{}
\DeclareMathAlphabet{\mathpzc}{OT1}{pzc}{m}{it}

\usepackage{hyperref}
\usepackage{amsmath}
\usepackage{amssymb}
\usepackage{mathtools}
\usepackage{bm}
\usepackage{cleveref}
\usepackage{tensor}
\usepackage{braket}
\usepackage{enumitem}
\usepackage{mhchem}
\usepackage{amsthm}
\usepackage{nccmath}
\usepackage{mathrsfs}
\usepackage{color}

\newcommand{\spc}{\quad \quad \quad}

\newcommand{\E}{{\mathbb{E}}}
\newcommand{\Sb}{{\mathbb{S}}}
\newcommand{\Pb}{{\mathbb{P}}}
\newcommand{\s}{{\sigma \!\! \! \sigma}}

\def\be{\begin{equation}}
\def\ee{\end{equation}}
\def\beq{\begin{eqnarray}}
\def\eeq{\end{eqnarray}}

\theoremstyle{definition}
\newtheorem{definition}{Definition}

\theoremstyle{theorem}
\newtheorem{theorem}{Theorem}

\theoremstyle{corollary}
\newtheorem{corollary}{Corollary}

\begin{document}
\title{How Lorentz boosts reshape relaxation spectra}
\author{L.~Gavassino}
\affiliation{Department of Applied Mathematics and Theoretical Physics, University of Cambridge, Wilberforce Road, Cambridge CB3 0WA, United Kingdom}

\begin{abstract}
In relativity, relaxation processes are often assumed to undergo time dilation under Lorentz boosts. We show that this intuition fails generically. Due to relativity of simultaneity, Lorentz boosts can split a single relaxation mode into a continuum of excitations, with a width set by the maximal signal propagation speed. Focusing on linearized relativistic (kinetic or rheological) theories with an Onsager-type symmetry, we derive rigorous bounds on relaxation spectra in arbitrary inertial frames, expressed solely in terms of rest-frame spectral data at zero wavenumber. As a consequence, non-hydrodynamic gaps, maximal relaxation rates, and the convergence radii of hydrodynamic modes obey nontrivial Lorentz-covariant constraints. These results provide a unified framework for understanding how relativity constrains relaxation dynamics in many-body systems.
\end{abstract} 
\maketitle

{\it \noindent \textbf{Introduction --}} Much of what we know about fluids follows from the study of their mode structure. Starting from a homogeneous equilibrium state, one introduces a small perturbation, and expands it in Fourier modes $e^{ikx - i\omega t}$ (with $k\in \mathbb{R}$). The resulting complex frequencies $\omega(k)$ determine how wavepackets propagate and decay, and thereby characterize the fluid’s response properties. Consequently, considerable effort has been devoted to rigorously analyzing and classifying mode spectra in hydrodynamics, kinetic theory, and holography \cite{KovtunHolography2005,Pu2010,Denicol_Relaxation_2011,Heller2014,RomatschkeCutsandPoles:2015gic,Grozdanov:2016vgg,Kurkela:2017xis,Moore:2018mma,Grozdanov:2019kge,Kovtun2019,Dash:2021ibx,GavassinoNonHydro2022,HellerHydrohedron2023jtd,Brants:2024wrx,RochaBranchcut:2024cge,GavassinoGapless:2024rck,Hu:2024tnn,Bajec:2024jez,Gavassino:2025tul,Bajec:2025dqm}.

Such analyses are usually carried out in the reference frame where the equilibrium state is at rest. In Newtonian physics, this entails no loss of generality, since the dispersion relations $\omega(k)$ and $\omega'(k')$ measured by two observers $\mathcal{O}$ and $\mathcal{O}'$ moving at relative velocity $v$ satisfy
\vspace{-0.3cm}
\begin{equation}
\begin{cases}
\omega'=\omega+vk \, ,\\
k'=k \, , 
\end{cases}
\qquad (\text{Galilei boost}).
\end{equation}
Since $k$ is invariant (and real, for Fourier modes), phases and group velocities simply shift by $v$, while the damping rate is untouched ($\mathfrak{Im}\,\omega'=\mathfrak{Im}\,\omega$).

In relativity, the picture changes qualitatively \cite{Hiscock_Insatibility_first_order,GavassinoSuperlum2021,Bhattacharyya:2025hjs,Hoult:2025htt}. Relativity of simultaneity causes $k'$ to depend on $\omega$:
\vspace{-0.2cm}
\begin{equation}
\begin{cases}
\omega'=\gamma(\omega+vk) \, ,\\[2pt]
k'=\gamma(k+v\omega) \, , 
\end{cases}
\qquad (\text{Lorentz boost}),
\end{equation}
with $\gamma\,{=}\,(1{-}v^2)^{-1/2}$. As a result, the mapping between $\omega(k)$ and $\omega'(k')$ is now implicit, and a mode that appears as a Fourier excitation to $\mathcal{O}$ (i.e. $k\in\mathbb{R}$) no longer corresponds to a Fourier mode for $\mathcal{O}'$ (i.e. $k'\notin\mathbb{R}$) \cite{GavassinoSuperlum2021}. This can change most qualitative features of the spectrum, even for the simplest systems. Consider, for instance, the diffusive mode $\omega=-ik^{2}$: a single-valued, stable dispersion relation (i.e. $\mathfrak{Im}\,\omega\le0$ for $k\in\mathbb{R}$~\cite{Hiscock_Insatibility_first_order,Kovtun2019}), gapless, and analytic with infinite radius of convergence around $k=0$. After a boost with, say, $v=1/2$, one instead finds two branches, $\omega'_{\pm}=2k'+\sqrt{3}\,i\bigl(1\pm\sqrt{1-\sqrt{3}\,i\,k'}\bigr)$, one of which is unstable and gapped, and both of which possess only a finite radius of convergence about $k'=0$.

Recently, the appearance of the unstable branch $\omega_+'$ was shown to be connected to the ability of the diffusion equation $\omega{=}{-}ik^{2}$ to propagate information superluminally \cite{GavassinoSuperlum2021}. It was further established that, if the underlying theory is causal, ``stability'' (or the lack thereof) is promoted to a Lorentz-invariant property \cite{GavassinoBounds2023myj}. This prompts a broader question: Are there additional, physically motivated assumptions (besides causality) that ensure Lorentz invariance of other spectral features, like the existence of gaps or the analyticity properties of $\omega(k)$?

In this Letter, we demonstrate that, when the equations governing relaxation to equilibrium obey a certain ``Onsager-type symmetry'' (which is a built-in feature of kinetic theory and rheology), a wide variety of spectral properties become Lorentz invariant. In particular, one can place bounds on both magnitudes and convergence radii of the frequencies measured by a moving observer solely in terms of rest-frame spectral features at $k=0$.

{\it \noindent \textbf{Conventions --}} We adopt the metric signature $(-,+,+,+)$, and work in natural units, $c=\hbar=k_B=1$.

{\it \noindent \textbf{The symmetry assumption --}}
Consider a homogeneous equilibrium state, and let $\Psi(x^\mu)\,{\in}\, \mathbb{C}^D$ (with $D\,{\in}\, \mathbb{N}$) be a list of linearized fields describing small perturbations about it. These may represent, for example, fluctuations of the moments of the kinetic distribution function\footnote{\label{footona1}For clarity, we work at finite $D$, which in kinetic theory amounts to truncating the moment expansion at some very large number (e.g. $10^{100}$ moments). This is physically acceptable: the distribution function is already a coarse-grained object, and kinetic theory ceases to hold at too fine momentum resolution \cite[\S 3.1]{huang_book}.}, or a set of transient hydrodynamic variables. We assume that their linear evolution is governed by a first-order system of coupled partial differential equations,
\vspace{-0.1cm}
\begin{equation}\label{EoM}
\E^\mu \partial_\mu \Psi = -\s \Psi \, ,
\end{equation}
where $\E^\mu$ and $\s$ are constant background matrices. Our key structural assumption is that all these matrices are \textit{Hermitian}: $\E^\mu\,{=}\,(\E^\mu)^\dagger$ and $\s\,{=}\,\s^\dagger$. Furthermore, we require that $\s$ be non-negative definite, and $\E^{0}$ be positive definite in \emph{every} inertial frame. It follows that the quadratic forms $\Psi^\dagger \E^\mu \Psi$ constitute a future-directed timelike (or null) vector, which implies causality \cite{GerochLindblom1990,GavassinoCausality2021} (if $\E^\mu$ were non-Hermitian, one would need to compute the characteristics explicitly, to determine causality).


The above assumptions seem rather restrictive, but are satisfied by a very broad class of matter models \cite{Geroch_Lindblom_1991_causal,GavassinoNonHydro2022,GavassinoSymmetric2022nff,GavassinoUniveraalityI2023odx}, including relativistic kinetic theory \cite{RochaGavassinoFlucut:2024afv,GavassinoGapless:2024rck}, radiation transport \cite{GavassinoRadiazione,GavassinoRadiativeBounds:2025bxx}, Israel–Stewart theory \cite{Israel_Stewart_1979,Hishcock1983,OlsonLifsh1990}, divergence-type theories \cite{Liu1986,GerochLindblom1990}, GENERIC \cite{Stricker2019,GavassinoGENERIC:2022isg}, Carter’s multifluid formulation \cite{carter1991,PriouCOMPAR1991,GavassinoStabilityCarter2022}, and theories of (super)solidity \cite{GavassinoUniveraalityI2023odx,GavassinoUniversalityII:2023qwl}. Indeed, the constraints on the structure of $\E^\mu$ and $\s$ can be traced back to Onsager's theory \cite{Onsager_1931,Onsager_Casimir,Krommes1993}: $\E^0$ and $\s$ are Hessians of the free-energy density and the entropy production rate respectively (and so are positive definite in all frames), while Hermiticity arises as a reciprocal relation related to $PT$ symmetry \cite{GavassinoSymmetric2022nff,GavassinoGapless:2024rck} (see Supplementary Material for more details).

We stress that, in the present work, we are mostly interested in models whose non-hydrodynamic sector is physical. Hence, frameworks like Israel-Stewart should be regarded here not just as causal extensions of hydrodynamics \cite{Geroch1995}, but as mesoscopic ``quasi-hydrodynamic'' \cite{Denicol_Relaxation_2011,Grozdanov2019,StephanovHydroPlus:2017ghc} models of non-Newtonian fluids \cite{landau7,Andrade:2019zey,BAGGIOLI20201,GavassinoBurgers:2023eoz,GavassinoFarFromBulk:2023xkt,AhnBaggioli:2025odk}.


{\it \noindent \textbf{Some basic invariant statements --}} 
To warm up, we derive a few elementary Lorentz-invariant features of the excitation spectrum that follow immediately from our assumptions. To this end, consider perturbations of the form $\Psi\,{\propto} \,e^{ik_\mu x^\mu}$, with complex $k^\mu$. Then, equation \eqref{EoM} reduces to
\vspace{-0.2cm}
\begin{equation}\label{EomFour}
k_\mu \E^\mu \Psi =i\s \Psi\, .
\end{equation}
Multiplying both sides by $\Psi^\dagger$, we obtain
\begin{equation}\label{omegone}
k_\mu \Psi^\dagger \E^\mu \Psi =i\Psi^\dagger \s \Psi\,  \, ,
\end{equation}
Taking the real and imaginary parts of \eqref{omegone}, and noting that $\Psi^\dagger \E^\mu \Psi$ and $\Psi^\dagger \s \Psi$ are real, we obtain
\vspace{-0.1cm}
\begin{equation}\label{ReandIm}
\begin{split}
(\Psi^\dagger \E^\mu \Psi) \mathfrak{Re} \, k_\mu ={}& 0 \, , \\
(\Psi^\dagger \E^\mu \Psi) \mathfrak{Im} \, k_\mu ={}& \Psi^\dagger \s \Psi \geq 0 \, .
\end{split}
\end{equation}
Since $\Psi^\dagger \E^\mu \Psi$ is timelike future-directed, it follows that 
$\mathfrak{Re}\, k^\mu$ is spacelike, and $\mathfrak{Im}\, k^\mu$ lies 
outside the future light cone \cite[\S 1.4.2]{special_in_gen} (it is either spacelike or past directed). 
To see what this implies in practice, take $k^\mu=(\omega,k,0,0)$. 
The real-part equation then gives 
$|\mathfrak{Re}\,\omega|\leq |\mathfrak{Re}\,k|$, 
which means that the phase velocity $\mathfrak{Re}\,\omega/\mathfrak{Re}k$ 
never exceeds the speed of light in any of these models (although the group velocity can). 
This behavior is well known from both hydrodynamics \cite{Pu2010} and kinetic theory 
\cite[\S 9.3.1]{cercignani_book}\footnote{In kinetic theory, one usually finds that, for each $k>0$, there is a 
continuum of ``improper'' solutions of \eqref{EomFour}, with 
$|\mathfrak{Re}\omega|$ extending up to $k$, but never beyond 
\cite{RomatschkeCutsandPoles:2015gic}.}. 
We also recover the standard result that $\omega\in i\mathbb{R}$ at $k=0$, 
a universal feature of kinetic and transient hydrodynamic theories \cite{GavassinoNonHydro2022}. Due to the manifest covariance of \eqref{ReandIm}, we see that these statements hold in all reference frames.

The imaginary-part equation instead yields $\mathfrak{Im}\,\omega \,{\leq}\, |\mathfrak{Im}\,k|$, 
which expresses the fact that these models are stable in all inertial frames 
\cite{GavassinoBounds2023myj,GavassinoDisperisons2023mad} and that their 
dispersion relations respect microcausality \cite{HellerBounds2022ejw}.

\newpage
{\it \noindent \textbf{Conservation laws --}} Before presenting our main result, we must introduce a few relevant concepts.

Let $\Phi_I$ ($I\,{=}\,1,2,...,N$) be a basis of $\ker(\s)$. Since $\s\,{=}\,\s^\dagger$ and $\s \Phi_I \,{=}\, 0$, it follows that 
$\Phi_I^\dagger \s\, {=}\, (\s \Phi_I)^\dagger \,{=}\, 0$. 
Contracting \eqref{EoM} with $\Phi_I^\dagger$ then yields a set of linearized local 
conservation laws:
\begin{equation}\label{conserviamo}
\partial_\mu J_I^\mu = 0 \, ,
\end{equation}
for the currents $J_I^\mu \,{=}\, \Phi_I^\dagger \E^\mu \Psi$.
This set is complete. 

Indeed, suppose that there exists another conserved current of the form 
$J^\mu = \mathbb{A}^\mu \Psi$, where 
$\mathbb{A}^\mu \in (\mathbb{C}^D)^\dagger$ are constant. 
Introduce $\mathbb{B}^\mu = \mathbb{A}^\mu (\E^0)^{-1}$, 
which is well defined because $\E^0$ is positive definite. 
Then,
\begin{equation}
\partial_\mu J^\mu = 0 
\quad \Rightarrow \quad
\mathbb{B}^\mu \E^0 \partial_\mu \Psi = 0 \, .
\end{equation}
Subtracting $\mathbb{B}^0 \times$\eqref{EoM}, we get
$
(\mathbb{B}^j \E^0 {-} \mathbb{B}^0 \E^j) \partial_j \Psi 
= \mathbb{B}^0 \s \Psi
$.
For this relation to hold for arbitrary initial data $\Psi(0,x^j)$, 
we must require $\mathbb{B}^0 \s \,{=}\, 0$, which implies
$\mathbb{B}^0 \,{=}\, \sum_I \lambda^I \Phi_I^\dagger$ 
for some constants $\lambda^I$, and 
$\mathbb{B}^j \E^0 \,{=}\, \mathbb{B}^0 \E^j$, which 
then gives
\begin{equation}
J^\mu = \mathbb{B}^0 \E^\mu \Psi 
= \sum\nolimits_I \lambda^I J_I^\mu \, .
\end{equation}
Hence, all linearized conserved currents are linear combinations 
of those generated by the basis $\Phi_I$.

{\it \noindent \textbf{Nearby equilibria --}}
The vectors $\Phi_I$ play an additional key role. One immediately verifies that any constant state of the form
\begin{equation}\label{equilibrium}
\Psi(x^\mu)=\sum\nolimits_I c^I \Phi_I , \qquad c^I=\text{const}\in\mathbb{C},
\end{equation}
solves \eqref{EoM}. Such configurations have the property that the associated conserved density
\begin{equation}
\sum\nolimits_I (c^I)^* J_I^0 = \sum\nolimits_I (c^I)^* \Phi_I^\dagger \E^0 \Psi = \Psi^\dagger \E^0 \Psi
\end{equation}
is nonzero for $\Psi\neq 0$. Since $\Psi$ parametrizes linear departures from the background equilibrium, the states in \eqref{equilibrium} may be viewed as nearby equilibria obtained by infinitesimally shifting the conserved densities.

{\it \noindent \textbf{Nonhydrodynamic spectrum --}} Fix an observer $\mathcal{O}$ with four-velocity $u^\mu$, and look for exponential
solutions $\Psi\,{\propto}\, e^{ik_\mu x^\mu}$ that are spatially uniform (i.e. have $k\,{=}\,0$) in the frame of
$\mathcal{O}$, so that $k^\mu=\omega u^\mu$, with $\omega$ the frequency measured
by $\mathcal{O}$. Substituting this ansatz into \eqref{EomFour} yields
\begin{equation}\label{Nonhydromodes}
\omega (-u_\mu \E^\mu)\Psi = -i\s \Psi 
\end{equation}
This can be solved by working on a basis in $\Psi$--space\footnote{For any invertible $D{\times} D$ matrix $\Sb$, we can write $\E^\mu=\Sb^\dagger\Tilde{\E}^\mu\Sb$ and $\s=\Sb^\dagger\Tilde{\s}\Sb$, where $\Tilde{\E}^\mu$ and $\Tilde{\s}$ are Hermitian and inherit the positivity properties of the originals. Plugging these decomposions into \eqref{EoM}, we obtain an equivalent theory $\Tilde{\E}^\mu \partial_\mu \Tilde{\Psi}{=}{-}\Tilde{\s}\Tilde{\Psi}$, with $\Tilde{\Psi}{=}\Sb \Psi$.} such that $-u_\mu \E^\mu = 1$ and $\s=\text{diag}(\sigma_n)$, with $\sigma_n\geq 0$. In this basis, \eqref{Nonhydromodes} reduces to $\s\Psi\,{=}\,i\omega\Psi$. The eigenmodes are therefore the unit vectors $\Psi_n{=}\,(0,...,0,1_{\text{at }n},0,...,0)^T$, with eigenfrequencies $\omega_n=-i\sigma_n$. Among them, $N$ modes have $\omega\,{=}\,0$, and correspond to the equilibria \eqref{equilibrium}. The remaining $D{-}N$ modes have $i\omega_n{=}\sigma_n{>}0$, and define the \textit{nonhydrodynamic spectrum} (at $k=0$) relative to $\mathcal{O}$.


\newpage
{\it \noindent \textbf{Main theorem --}} We can finally state our main result:
\begin{theorem}\label{theo1}
Consider two inertial observers $\mathcal{O}$ and $\mathcal{O}'$ moving at relative speed $v\,{\geq}\,0$. Suppose that the nonhydrodynamic spectrum (at $k\,{=}\,0$) relative to $\mathcal{O}$ is confined within some interval $a\, {\leq}\, i\omega\,{\leq}\, b$ (with $a$ non-negative). Then, the nonhydrodynamic spectrum (at $k'{=}0$) relative to $\mathcal{O}'$ is confined within the interval
\begin{equation}\label{theboundoni}
\boxed{\dfrac{a(1{-}v)}{\gamma}\leq i\omega' \leq \dfrac{b}{\gamma (1{-}v)}  \, .}
\end{equation}
\end{theorem}
\begin{proof}
We carry out the analysis in the rest frame of $\mathcal{O}$, and choose a basis in $\Psi$--space such that $\E^0=1$ in such a frame. Then, the nonhydrodynamic modes relative to $\mathcal{O}$ fulfill the equation
\begin{equation}\label{eigenizzo}
i\omega \Psi =\s \Psi \qquad (\text{with }i\omega>0).    
\end{equation}
Multiplying both sides by $\Phi_I^\dagger$ ($I=1,...,N$), we find that $\Phi_I^\dagger \Psi \,{=}\, 0$. Hence, the nonhydrodynamic frequencies are the eigenvalues of $\s$ restricted to $(\ker \s)^{\perp}$.
A theorem of spectral theory \cite[Th. 2.19]{TeschlBook} states that the extremal values of $\Psi^\dagger \s \Psi/(\Psi^\dagger\Psi)$ are eigenvalues of $\s$, so that
\vspace{-0.1cm}
\begin{equation}\label{inebriante}
\begin{split}
\inf_{(\ker \s)^{\perp}} \frac{\Psi^\dagger \s \Psi}{\Psi^\dagger \Psi} =\inf_{(\ker \s)^{\perp}} \text{Spectrum}(\s) \geq a \, ,\\
\sup_{(\ker \s)^{\perp}} \frac{\Psi^\dagger \s \Psi}{\Psi^\dagger \Psi} =\sup_{(\ker \s)^{\perp}} \text{Spectrum}(\s) \leq b \, .
\end{split}
\end{equation}
Now, if we orient the axes so that the four-velocity of $\mathcal{O}'$ is
$u'^\mu =(\gamma,-\gamma v,0,0)$, then the nonhydrodynamic modes relative to $\mathcal{O}'$ fulfill the equation $i\gamma\omega' (1{+}v \E^1)\Psi = \s \Psi$ (with $\omega'>0$). Multiplying both side by either $\Psi^\dagger$ or $\Phi_I^\dagger$ ($I=1,...,N$), we obtain
\begin{equation}\label{Nonhydromodes1}
i\gamma\omega' = \frac{\Psi^\dagger \s \Psi}{1+v \Psi^\dagger  \E^1\Psi} \, , \quad \Phi_I^\dagger \Psi = -v\Phi_I^\dagger \E^1 \Psi \, ,
\end{equation}
where we have normalized $\Psi$ such that $\Psi^\dagger \Psi=1$. Since $\Psi^\dagger \E^\mu \Psi$ is timelike or null, we have that $|\Psi^\dagger \E^1\Psi|\leq \Psi^\dagger \E^0 \Psi=1$. Hence, the denominator in \eqref{Nonhydromodes1} is bounded between $1-v$ and $1+v$. To bound the numerator, we take $\Phi_I$ to be orthonormal, and define the matrix $\Pb=1-\sum_I \Phi_I \Phi_I^\dagger$, which is the orthogonal projector ($\Pb^2{=}\Pb^\dagger{=}\Pb$) onto $(\ker\s)^{\perp}$. Clearly, $\Psi^\dagger \s\Psi=\Psi^\dagger\Pb^\dagger \s \Pb\Psi$, so we can use \eqref{inebriante} to derive the following inequalities:
\begin{equation}
a\, \Psi^\dagger \Pb \Psi \leq \Psi^\dagger \s \Psi \leq b \, \Psi^\dagger \Pb \Psi \leq b\, . 
\end{equation}
The upper bound in \eqref{theboundoni} is immediately recovered. To obtain the lower bound, we invoke the definition of $\Pb$ and the second condition in \eqref{Nonhydromodes1}:
\begin{equation}\label{patatuzza}
\begin{split}
\Psi^\dagger \Pb \Psi ={}& 1{-}\sum_I |\Phi_I^\dagger \Psi|^2= 1{-}v^2\sum_I |\Phi_I^\dagger \E^1\Psi|^2 \\
\geq {}& 1-v^2 ||\E^1 \Psi ||^2  \geq 1{-}v^2 \, .\\ 
\end{split}   
\end{equation}
Here, we used Bessel's inequality and the spectral norm identity $||\E^1||\equiv \sup\limits_{||\Phi||=1} |\Phi^\dagger \E^1\Phi|\leq 1$ \cite[Eq. (2.65)]{TeschlBook}.
\end{proof}
Note that, to simplify the presentation, we have been working with a finite-dimensional $\Psi$ (see footnote \ref{footona1}). The argument, however, extends directly to the case where $\Psi$ lives in a Hilbert space, up to standard functional–analytic subtleties \cite{Bhatiabook,TeschlBook}. For illustration, the Supplementary Material provides a complete proof of Theorem \ref{theo1} within linearized kinetic theory.

{\it \noindent \textbf{Two quick examples --}} To test the bounds, we examine two representative models with $a=b=1$, in some units.  
Our first example is Cattaneo’s theory of diffusion \cite{cattaneo1958}, which in the rest frame takes the form
\begin{equation}
\begin{cases}
\partial_t T +w\partial_j q^j =0 \, ,\\
\partial_t q^j +w \partial^j T =-q^j \, , \\
\end{cases}
\end{equation}
and satisfies all of our hypotheses, provided $|w|\,{\leq}\,1$. In a boosted frame, this model features three non-hydrodynamic modes, with eigenfrequencies
\vspace{-0.2cm}
\begin{equation}\label{cattaniamo}
i\omega'_1 =\dfrac{1}{\gamma(1{-}v^2 w^2)} \, , \qquad i\omega'_{2,3}=\dfrac{1}{\gamma}\, .
\end{equation}
Our second example is kinetic theory in the Relaxation Time Approximation (RTA),
$p^\mu \partial_\mu \delta f =p^\mu u_\mu (\delta f-\delta f_{\text{eq}})$,
where $\delta f(x^\mu,p^\alpha)$ is the perturbed distribution function. In a boosted frame, this model exhibits a continuous family of non-hydrodynamic excitations of the form
\vspace{-0.2cm}
\begin{equation}\label{RTAmo}
i\omega'=\dfrac{1}{\gamma (1{+}vw)} \qquad (\text{with }w\,{\in}\, [-1,1]) \, .
\end{equation}
This splitting arises due to relativity of simultaneity: In the rest frame, particles travel different intervals $\Delta x\,{=}\,w\Delta t$ in the same amount of time $\Delta t \,{\sim}\, \tau$. When we boost, we have $\Delta t'\,{=}\,\gamma (\Delta t {+}v\Delta x)\,{\sim}\, \gamma (1{+}vw)\tau=1/(i\omega')$.

In figure \ref{fig:splitting}, we graph the modes of these two models against bounds \eqref{theboundoni}. RTA saturates the upper bound. We are not aware of a theory saturating the lower bound.

\begin{figure}[b!]
    \centering
\includegraphics[width=0.915\linewidth]{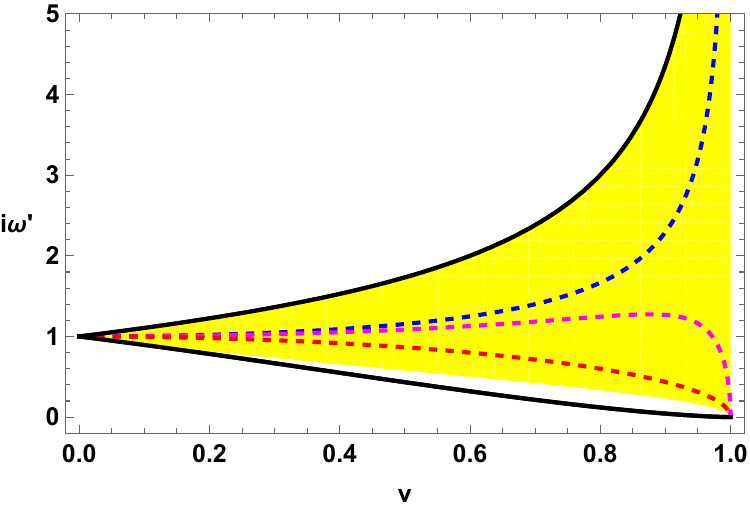}
\caption{Splitting of a single non-hydrodynamic frequency $i\omega\,{=}\,1$ into multiple frequencies $i\omega'$ induced by a boost of speed $v$. The black curves show the bounds \eqref{theboundoni}, with $a=b=1$, and they delimit the allowed region. The dashed curves refer to Cattaneo's eigenfrequency $i\omega'_1$, see equation \eqref{cattaniamo}, with $w=0$ (red), $0.9$ (magenta) and $1$ (blue); the red curve also captures the modes $i\omega'_{2,3}$. The yellow region covers the continuum of modes \eqref{RTAmo} of RTA.
}
    \label{fig:splitting}
\end{figure}

\newpage
{\it \noindent \textbf{Some applications --}} We now discuss some important consequences of Theorem~\ref{theo1}, starting with two corollaries that are particularly relevant for kinetic theory.

\begin{corollary}\label{cor1}
If the non-hydrodynamic spectrum possesses a fastest relaxation timescale at $k=0$ in one inertial frame, then the same is true in all inertial frames.
\end{corollary}

\begin{corollary}\label{cor2}
If the non-hydrodynamic spectrum is gapped at $k=0$ in one inertial frame, then the same is true in all inertial frames.
\end{corollary}

Both results follow directly from Theorem~\ref{theo1}, by choosing $a=0$ and finite $b$ in the first case, and finite $a$ with $b=\infty$ in the second.

Corollary~\ref{cor2} has an important implication. In \cite{GavassinoConvergence:2024xwf} it was shown that, in kinetic theory, a gap in the non-hydrodynamic spectrum guarantees that all hydrodynamic dispersion relations admit expansions as power series in $k$ with a finite radius of convergence.
The proof does not rely on the background being at rest, and therefore the result persists under boosts. Combining this with Corollary~\ref{cor2}, we obtain the following statement:

\begin{corollary}
If the non-hydrodynamic spectrum of a kinetic theory is gapped at $k=0$ in one inertial frame, then the hydrodynamic dispersion relations possess non-vanishing convergence radii in all inertial frames.
\end{corollary}

This means that, if a genuine separation of scales exists, density-frame hydrodynamics \cite{ArmasDensityFrame2020mpr,Basar:2024qxd,BhambureDensityFrame2024axa} is well defined up to infinite order in derivatives in every inertial frame\footnote{Density-frame hydrodynamics posits that, in any inertial frame (not necessarily comoving), the conserved fluxes $J_I^k$ admit an expansion in powers of \emph{spatial} derivatives of the conserved densities $J_I^0$. Equivalently, that the hydrodynamic eigenfrequencies $\omega$ can be expanded in powers of the spatial wavevector $k^j$.}.

{\it \noindent \textbf{Convergence radius of boosted hydrodynamics --}} Consider a kinetic description of a system admitting a single conserved current, $\partial_\mu J^\mu = 0$. 
Assume that its nonhydrodynamic spectrum is gapped in the rest frame, with slowest 
relaxation rate $a$. In a frame boosted with velocity $v$, the corresponding spectral gap 
is $a' \geq a(1{-}v)/\gamma$. In that boosted frame, the unique hydrodynamic mode $\omega'(k')$ propagating along the 
$x$--direction solves
\begin{equation}
(\s + \E^1 ik')\,\Psi = i\omega'\Psi \, ,
\end{equation}
where we work on a basis for which $\E^0 = 1$. Upon complexifying $k'$, this equation 
defines an eigenvalue problem for the perturbed family $\s(ik') \equiv \s + ik'\E^1$.
Since both $\s$ and $\E^1$ are Hermitian (with $||\E^1||\leq 1$), and $0$ is a non-degenerate eigenvalue of $\s$, standard results from analytic perturbation 
theory apply \cite[\S II.3.5]{Kato_Perturbation_Theory}. In particular,
there exists a unique, analytic hydrodynamic mode $\{ \omega'(k'), \Psi(k') \}$ 
in a neighborhood of $k'=0$, whose radius of convergence $R'$ satisfies the lower bound \cite{Kato1949_Perturbation_I}
\begin{equation}
\boxed{R' \ge \frac{a'}{2\,||{\E^1}||} \ge \frac{a(1-v)}{2\gamma} \, .}
\end{equation}

\newpage
{\it \noindent \textbf{About time dilation --}} It is natural to expect that, as the velocity of the fluid approaches the speed of light, all nonhydrodynamic modes should effectively ``freeze'' due to time dilation. From this viewpoint, one would predict that the upper bound on $i\omega'$ should vanish as $v\!\to\!1$. Figure \ref{fig:splitting} shows that the opposite happens: the upper bound diverges. Where does our intuition fail?

The time-dilation argument implicitly assumes that signals are advected with the fluid element. Indeed, if, in the rest frame, a signal propagates from $(t,x)=(0,0)$ to $(t,x)=(\Delta t,0)$ before decaying, then in a frame where the medium moves at velocity $v$, the signal travels from $(t',x')=(0,0)$ to $(t',x')=(\gamma \Delta t,\gamma v \Delta t)$. Under this assumption, one finds that $i\omega'\sim 1/(\gamma\Delta t)$, which tends to zero as $v\to 1$, in line with expectations.
The problem appears if the signal has a nonzero rest-frame propagation speed $w$. If in the rest frame it moves from $(t,x)=(0,0)$ to $(t,x)=(\Delta t,- w \Delta t)$, then in the boosted frame it travels to
$(t',x')=(\gamma \Delta t- \gamma v w \Delta t,\, - \gamma w \Delta t+\gamma v \Delta t)$, which leads to $i\omega' \sim 1/[\gamma(1-vw)\Delta t]$.
Differentiating with respect to $v$ then shows that, for $v<w$, the boosted frequency \emph{increases} with $v$, in sharp contrast with the simple time-dilation picture. Only once $v>w$ does time dilation dominate, forcing the frequency downward. In the limiting case of a luminal signal ($w=1$), we recover $i\omega' \sim 1/[\gamma(1-v)\Delta t]$, which has the same $v$--dependence as the upper bound in \eqref{theboundoni}, and diverges as $v\rightarrow 1$

The argument above suggests that, if the medium supports signal propagation
only inside an ``acoustic cone'' narrower than the light cone, then
the boosted nonhydrodynamic spectrum should obey bounds sharper than
\eqref{theboundoni}. This is made precise in the theorem below.
\vspace{-0.1cm}
\begin{theorem}\label{theo2}
Consider a medium whose $k{=}0$ nonhydrodynamic spectrum lies within some interval $a \le i\omega \le b$ (with $a\ge 0$) in the rest frame. Assume further that signal propagation in the rest frame is bounded by a maximal speed $w\,{\le}\, 1$. Then, for an observer moving at speed $v\,{\geq}\,0$ relative to the medium, the corresponding $k'{=}0$ nonhydrodynamic frequencies are confined within the interval 
\vspace{-0.3cm}
\begin{equation}\label{theboundoni2} 
\boxed{\frac{a(1{-}vw)}{\gamma} \leq i\omega' \leq \frac{b}{\gamma(1{-}vw)}\, .} \end{equation} 
\end{theorem}
\begin{proof}
The characteristic propagation speeds $v_n$ of signals traveling along the
$x$--direction are determined by the condition
$\det(\E^1 - v_n \E^0)=0$ \cite{Hishcock1983}. In a basis where $\E^0 = 1$,
the $v_n$ coincide with the eigenvalues of $\E^1$, so imposing that no signal
exceed the maximal speed $w$ is equivalent to the operator bound
$||\E^1|| \le w$. With this identification, the derivation of
Theorem~\ref{theo1} carries over unchanged when $\mathcal{O}$ is chosen to
comove with the medium. The only modification is that, in
Eq.~\eqref{Nonhydromodes1}, the denominator is now bounded between
$1 - v w$ and $1 + v w$. Likewise, in the final step of
\eqref{patatuzza}, we can replace $1 - v^2$ with $1 - v^2 w^2$.
\end{proof}
In figure \ref{fig:tighter}, we plot the bounds \eqref{theboundoni2} with $a=b=1$, for various values of $w$. In the Supplementary Material, we test such bounds with some randomly-generated models.

\begin{figure}[h!]
    \centering
\includegraphics[width=0.95\linewidth]{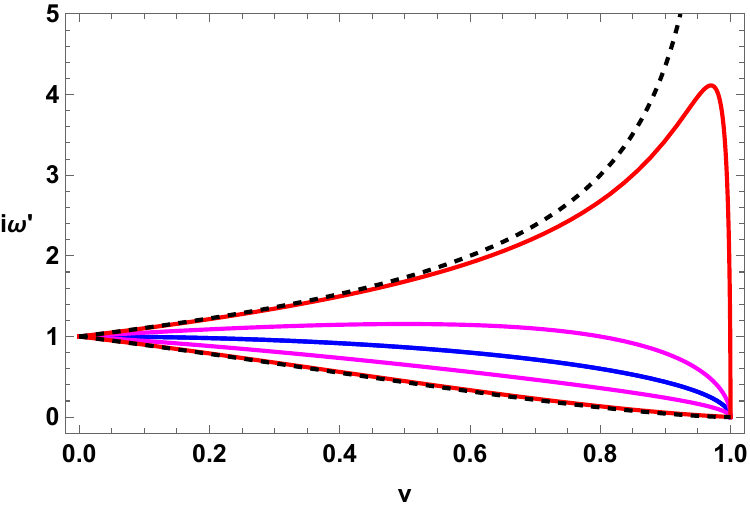}
\caption{Bounds \eqref{theboundoni2} on the nonhydrodynamic spectrum of a moving fluid, assuming a single rest-frame nonhydrodynamic frequency
$i\omega = 1$, and rest-frame signaling speed
$w=0$ (blue), $0.5$ (magenta), $0.97$ (red), and $1$ (black dashed). For
$w=0$, the admissible region collapses to the time--dilation line $i\omega' = 1/\gamma$,
while for $w=1$ the result reduces to \eqref{theboundoni}, which is the only
case in which $i\omega'$ may diverge as $v \to 1$.
}
    \label{fig:tighter}
\end{figure}

{\it \noindent \textbf{Non-relativistic media --}} In the regime where the microscopic dynamics are nonrelativistic in the rest frame (so that $w\!\ll\!1$), the usual time--dilation intuition is recovered: the bounds collapse to $a/\gamma \le i\omega' \le b/\gamma$. Indeed, in this limit, the boosted eigenvalue equation $i\gamma\omega'(1{+}v \E^1)\Psi = \s\Psi$ simplifies to $i\gamma\omega'\Psi = \s\Psi$, implying that $\gamma\omega'$ matches a corresponding rest--frame nonhydrodynamic frequency.

\newpage

{\it \noindent \textbf{Conclusions --}} We have set rigorous foundations for a relativistically covariant spectral theory of moving fluids, which applies to a wide variety of substances, including gases, solids, superfluids, and supersolids \cite{GavassinoUniveraalityI2023odx,GavassinoUniversalityII:2023qwl}. The resulting physics departs sharply from naive time-dilation expectations: In fully relativistic media, the equilibration timescale of some modes can become infinitely fast under large boosts (see figure \ref{fig:splitting}). This divergence, however, is avoided if the medium has a strictly subluminal signaling speed $w$ (see figure \ref{fig:tighter}). Moreover, the usual time-dilation intuition is recovered in substances that are non-relativistic in the rest frame, for which $w\ll 1$.

Our results point to several promising directions for future investigation. A particularly natural extension concerns the applicability of quasi-hydrodynamic \cite{StephanovHydroPlus:2017ghc,Grozdanov2019,BulkGavassino,GavassinoFarFromBulk:2023xkt,AhnBaggioli:2025odk} at large boosts. One can show that there exist systems where a parametrically slow mode is absorbed into the continuum past a critical velocity (see Supplementary Material), raising the question of which additional assumptions are required for quasi-hydrodynamics to exist in all reference frames. Another important question is whether the bounds derived here extend to electrically charged fluids or to holographic quasi-normal modes, neither of which obey the symmetry principle assumed in this work \cite{GavassinoNonHydro2022,Gavassino:2025tul}. It is tempting to conjecture that conclusions such as Corollaries~\ref{cor1} and~\ref{cor2} remain valid in these broader settings, but establishing them without symmetry constraints is likely to require substantially more sophisticated techniques.

\vspace{-0.4cm}
\section*{Acknowledgements}
\vspace{-0.4cm}

This work is supported by a MERAC Foundation prize grant,  an Isaac Newton Trust Grant, and funding from the Cambridge Centre for Theoretical Cosmology.

\bibliography{Biblio}

\newpage

\onecolumngrid
\newpage
\begin{center}
\textbf{\large How Lorentz boosts reshape relaxation spectra:\\Supplementary Material}\\[.2cm]
  L. Gavassino\\[.1cm]
  {\itshape Department of Applied Mathematics and Theoretical Physics, University of Cambridge, Wilberforce Road, Cambridge CB3 0WA, United Kingdom\\}
\end{center}

\setcounter{equation}{0}
\setcounter{figure}{0}
\setcounter{table}{0}
\setcounter{page}{1}
\renewcommand{\theequation}{S\arabic{equation}}
\renewcommand{\thefigure}{S\arabic{figure}}

\section*{Derivation of the symmetry conditions from the Onsager-Casimir principle}\label{AppAAA}
\vspace{-0.3cm}

Fix a reference frame, and \textit{define} the matrix $\E^0$ so that, in the Gaussian approximation \cite[\S111]{landau5}, the fluctuation probability takes the standard form $\mathcal{P}\propto e^{-E}$, with
\begin{equation}
E=\frac{1}{2}\int \Psi^{T}\E^0\Psi\, d^{3}x .
\end{equation}
For physical states, we take $\Psi$ real, allowing $\E^0$ to be chosen real, symmetric, and (by stability) positive definite. 

Next, consider a wavevector $\mathbf{k}$, and align the coordinate axes so that $\mathbf{k}=(k,0,0)$. The corresponding fluctuating Fourier mode can then be decomposed as $\Psi(t,x,y,z)=\Psi_s(t)\sin(kx)+\Psi_c(t)\cos(kx)$, and its probability distribution (for real $\Psi_s$ and $\Psi_c$) reads
\begin{equation}
\mathcal{P}(\Psi_s,\Psi_c) = \dfrac{e^{-\frac{V}{4}(\Psi^T_s \E^0 \Psi_s+\Psi^T_c \E^0 \Psi_c)}}{\int e^{-\frac{V}{4}(\Psi^T_s \E^0 \Psi_s+\Psi^T_c \E^0 \Psi_c)} d\Psi_s d\Psi_c} \, , 
\end{equation}
with $V$ denoting the volume of the system. The equal-time correlator matrix of this Fourier mode is therefore
\begin{equation}
\mathbb{Q} \equiv
\begin{bmatrix}
\langle \Psi_s \Psi_s^T \rangle & \langle \Psi_s \Psi_c^T \rangle \\
\langle \Psi_c \Psi_s^T \rangle & \langle \Psi_c \Psi_c^T \rangle \\
\end{bmatrix}
= \dfrac{2}{V}
\begin{bmatrix}
(\E^0)^{-1} & 0 \\
0 & (\E^0)^{-1}\\
\end{bmatrix} \, .
\end{equation}

Assume now that $\Psi$ evolves according to a first-order system $\partial_t \Psi\,{=}\,{-}(\mathbb{M}_0{+}\mathbb{M}_1 \partial_x)\Psi$, with $\mathbb{M}_0$ and $\mathbb{M}_1$ some constant background matrices. For the Fourier mode $\Psi\,{=}\,\Psi_s(t)\sin(kx){+}\Psi_c(t)\cos(kx)$, the dynamics reduce to
\begin{equation}
\dfrac{d}{dt}
\begin{bmatrix}
\Psi_s \\
\Psi_c \\
\end{bmatrix}
=-
\begin{bmatrix}
\mathbb{M}_0 & -k\mathbb{M}_1 \\
k\mathbb{M}_1 & \mathbb{M}_0 \\
\end{bmatrix}
\begin{bmatrix}
\Psi_s \\
\Psi_c \\
\end{bmatrix} \equiv - \mathbb{M} \begin{bmatrix}
\Psi_s \\
\Psi_c \\
\end{bmatrix} \, .
\end{equation}

The Onsager--Casimir principle states that, if there exists a discrete transformation $\epsilon:(\Psi_s,\Psi_c)\rightarrow (\epsilon_s\Psi_s,\epsilon_c\Psi_c)$ that contains time reversal (such as T, PT, or CPT) which (i) is a symmetry of the microscopic theory and (ii) leaves the equilibrium state invariant, then \cite{Krommes1993}
\begin{equation}\label{Kromm}
\mathbb{M}\mathbb{Q}=\epsilon (\mathbb{M}\mathbb{Q})^T \epsilon^T \, .
\end{equation}
For our analysis, PT is the natural discrete symmetry to impose. Unlike time reversal alone, which flips the momentum of a moving state, PT leaves both energy and momentum unchanged, and therefore remains a symmetry of (unmagnetized) equilibria even in boosted frames \cite{PeiCovariantFluctuations2025}. Although PT is violated by the weak interaction, CPT cannot be used in its place, since CPT reverses the sign of the chemical potential and thus does not preserve most equilibria. 

Under PT, the distribution transforms as $f(x,p)\rightarrow f(-x,p)$ \cite{Bender:2023cem,cavaglia2015nonsemilinear}. Consequently, if the degrees of freedom are the moments of $f$ (or some hydrodynamic fluxes), we always have that $\Psi_s$ is PT-odd and $\Psi_c$ is PT-even, leading to
\begin{equation}
\epsilon =
\begin{bmatrix}
-1 & 0 \\
0 & 1 \\
\end{bmatrix} \, ,
\end{equation}
with ``$1$'' denoting the identity matrix in $\Psi$--space. Substituting this into \eqref{Kromm} yields
\begin{equation}
\E^0 \mathbb{M}_0=(\E^0 \mathbb{M}_0)^T \equiv \s \, , \spc \E^0 \mathbb{M}_1=(\E^0 \mathbb{M}_1)^T \equiv \E^1 \, . 
\end{equation}
Multiplying the equations of motion $\partial_t\Psi=-(\mathbb{M}_0+\mathbb{M}_1\partial_x)\Psi$ by $\E^0$ then produces $(\E^0\partial_t+\E^1\partial_x)\Psi=-\s\Psi$, where all matrices are Hermitian (in fact real symmetric). The matrix $\E^0$ is positive definite in every reference frame. The final property to establish is the non-negativity of $\s$, which follows directly from stability considerations.

The above argument shows that the Hermitian structure of $\E^j$ and $\s$ stems from the fact that the relevant kinetic and transient variables are PT-even. This is often true, but not always. There exist systems whose fundamental degrees of freedom are PT-odd. Familiar examples include the electric field, which is odd under parity and even under time reversal, and the magnetic field, which exhibits the opposite pattern \cite[\S 6.10]{jackson_classical_1999}. A further instance is provided by the quasi-normal modes of holographic fluids. Here, the gravitational equations in the bulk of the holographic dual description are of second order in time, which cause the non-hydrodynamic degrees of freedom to always come in pairs $\{\phi,\partial_t \phi\}$ that acquire opposite signs under time reversal \cite{GavassinoNonHydro2022}.

\newpage
\section*{Proof of Theorem 1 within kinetic theory}

\subsection{Mathematical set up}

We focus on a non-degenerate gas of spinless particles, and write the kinetic distribution function as $f=f_{\text{eq}}(1+\psi)$, where $f_{\text{eq}}(p^j)$ is a uniform equilibrium state, and $\psi(x^\mu,p^j)$ is a linear perturbation. Working in the rest frame of $\mathcal{O}$ (who need not comove with $f_{\text{eq}}$), and assuming $\psi\propto e^{i\omega t-ik_j x^j}$, the linearized Boltzmann equation reads
\begin{equation}\label{Boltzmann}
\left(i\omega-w^j ik_j\right)\psi =I \psi \, , 
\end{equation}
where $I$ is the (Friedrichs extension \cite[\S 2.3]{TeschlBook} of the) linearized collision operator\footnote{In some references, one writes $I=-\frac{1}{p^0}L$ \cite{RochaBranchcut:2024cge,GavassinoGapless:2024rck}.}, and $w^j=p^j/p^0$. Working within the Hilbert space $\mathcal{H}=L^2(\mathbb{R}^3,f_{\text{eq}}d^3 p)$, with inner product
\begin{equation}
(\phi,\psi)=\int \dfrac{d^3 p}{(2\pi)^3} f_{\text{eq}} \phi^* \psi  \, ,
\end{equation}
and norm $||\psi||=\sqrt{(\psi,\psi)}$, and assuming a non-vanishing cross section, one finds the following properties \cite{GavassinoConvergence:2024xwf}:
\begin{equation}
\begin{split}
&I\psi=0 \text{ if and only if }\psi\in \text{span}(1,p^\nu) \, , \\
&(\phi,I\psi)^*=(\psi, I\phi) \, , \\
&(\psi,I\psi)\geq 0 \, , \\
&(\phi,w^j\psi)^*=(\psi,w^j\phi) \, , \\
&(\psi,w^j \psi)\in [-||\psi||^2,||\psi||^2] \, . \\
\end{split}
\end{equation}
Thus, the structure mirrors that of the main text, under the replacements $\Psi \rightarrow \psi$, $\E^0\rightarrow 1$, $\E^j \rightarrow w^j$, $\s \rightarrow I$, $\Phi^\dagger \Psi \rightarrow (\phi,\psi)$, and $\{\Phi_I\}\rightarrow \{1,p^0,p^1,p^2,p^3\}$. Indeed, if we decompose $\psi$ into an orthonormal basis, and denote the list of linear combination coefficients by $\Psi$, we recover precisely the formalism of the main text, with $D=\infty$. 

There is, however, one caveat. In an infinite-dimensional Hilbert space, the excitation spectrum generally includes both point and continuous components, and equation \eqref{Boltzmann} only captures the former. To characterize the full spectrum, we adopt the following broader definition \cite[Lecture 18, Th. 1]{Bhatiabook}:
\begin{definition}\label{def1}
A wavevector $k^\mu =(\omega,k^j)$ is an eigenmode of the theory if there exists a sequence $\{\psi_n\}_{n=1}^{\infty}$ of normalized states (i.e. $||\psi_n||=1$) such that
\begin{equation}\label{limitone}
||\left(I-i\omega+w^j ik_j\right)\psi_n|| \rightarrow 0 \, .  
\end{equation}
\end{definition}

\subsection{Conservation laws and non-hydrodynamic modes}

In linearized kinetic theory, the perturbations of the conserved currents associated with a state $\psi$ take the form
\begin{equation}
\delta J^\mu = \int \dfrac{d^3 p}{(2\pi)^3} f_{\text{eq}} \dfrac{p^\mu}{p^0} \phi =
\begin{bmatrix}
(1,\psi) \\
(w^1,\psi) \\
(w^2,\psi) \\
(w^3,\psi) \\
\end{bmatrix} \, ,
\qquad 
\delta T^{\mu\nu} = \int \dfrac{d^3 p}{(2\pi)^3} f_{\text{eq}} \dfrac{p^\mu}{p^0} p^\nu \phi =
\begin{bmatrix}
(p^\nu,\psi) \\
(w^1 p^\nu,\psi) \\
(w^2 p^\nu,\psi) \\
(w^3 p^\nu,\psi) \\
\end{bmatrix} \, ,
\end{equation}
which represent the particle-number current and the stress-energy tensor. For a generic collision cross section, these exhaust the set of conserved currents \cite[\S 2.4]{cercignani_book}. In line with the definition used in the main text, we therefore call an eigenmode $k'^\mu=\omega' u'^\mu$ non-hydrodynamic relative to an observer $\mathcal{O}'$ (with four-velocity $u'^\mu$) if the associated state $\psi$ (or the sequence of states $\psi_n$ of Definition~\ref{def1}) carries no conserved density, i.e.
\begin{equation}
u'_\mu \delta J^\mu =u'_\mu \delta T^{\mu \nu}=0 \, .
\end{equation}

There is, however, a key distinction from the finite-dimensional setting. In that case, $\omega'\neq 0$ is required, since $\omega'=0$ would imply $\Psi\in\ker(\s)$, meaning that the state is a genuine equilibrium. In the infinite-dimensional kinetic-theory case, it may happen that none of the $\psi_n$ lies in $\ker(I)$, while the limit \eqref{limitone} holds with $k^\mu=0$. In that case, $\omega'=0$ belongs to the continuous non-hydrodynamic spectrum, which is then said to be gapless \cite{GavassinoGapless:2024rck}. For our purposes, this subtlety is immaterial, since the theorem already allows the lower endpoint of the interval $a\leq i\omega \leq b$ to satisfy $a=0$.

\subsection{Bounds on the collision operator}
\vspace{-0.2cm}

From the definitions above, a frequency $\omega$ lies in the $k=0$ non-hydrodynamic spectrum relative to $\mathcal{O}$ if there exists a sequence of states $\psi_n$ satisfying
\begin{equation}
||\psi_n||=1 \, , \quad \qquad (1,\psi_n)=(p^\nu,\psi_n)=0 \, , \quad \qquad ||\left(I-i\omega\right)\psi_n|| \rightarrow 0 \, . 
\end{equation}
Equivalently, $i\omega$ belongs to the spectrum of $I$ restricted to the Hilbert subspace $\{1,p^\nu\}^{\perp}$. Using \cite[\S 2.4, Th.~2.19]{TeschlBook} together with the assumed bound $a\leq i\omega\leq b$, we obtain the variational inequalities
\begin{equation}\label{bounduzzo}
\begin{split}
& \inf_{\{1,p^\nu\}^{\perp}} \dfrac{(\psi,I\psi)}{(\psi,\psi)} = \inf_{\{1,p^\nu\}^{\perp}} \text{Spectrum}(I)\geq a \, , \\
& \sup_{\{1,p^\nu\}^{\perp}} \dfrac{(\psi,I\psi)}{(\psi,\psi)} = \sup_{\{1,p^\nu\}^{\perp}} \text{Spectrum}(I)\leq b \, ,\\
\end{split}
\end{equation}
which directly parallel the bounds on $\Psi^\dagger \s\Psi/(\Psi^\dagger \Psi)$ established in the main text.

\vspace{-0.2cm}
\subsection{Bounds on boosted spectra}
\vspace{-0.2cm}

We are now ready to derive bounds on the $k'=0$ non-hydrodynamic spectrum relative to an observer $\mathcal{O}'$ that moves with speed $v$ relative to $\mathcal{O}$. The algebra is essentially identical to that of the main text, modulo some technicalities.

As usual, we orient the axes such that the fourvelocity of $\mathcal{O}'$ is $u'^\mu=(\gamma,-\gamma v,0,0)$. Then, $\omega'$ is a non-hydrodynamic frequency provided that there exists a sequence of states $\psi_n$ with the following properties:
\begin{equation}\label{gringooogrigoo}
||\psi_n||=1 \, , \quad \qquad (1{+}vw^1,\psi_n)=([1{+}vw^1] p^\nu,\psi_n)=0 \, , \quad \qquad ||\left[I-i\gamma\omega'(1{+}vw^1)\right]\psi_n|| \rightarrow 0 \, . 
\end{equation}
Using the Cauchy–Schwarz inequality, we have that
\begin{equation}
|(\psi_n,\left[I-i\gamma\omega'(1{+}vw^1)\right]\psi_n)|\leq ||\left[I-i\gamma\omega'(1{+}vw^1)\right]\psi_n|| \rightarrow 0 \, ,
\end{equation}
which implies
\begin{equation}
[1-v(\psi_n,w^1\psi_n)]\bigg[\dfrac{(\psi_n,I\psi_n)}{1-v(\psi_n,w^1\psi_n)}-i\gamma \omega'\bigg]\rightarrow 0\, .
\end{equation}
But since $|(\psi_n,w^1\psi_n)|\leq ||\psi_n||^2=1$, the first square bracket is larger $1-v$, and thus cannot tend zero. Hence, the second square bracket must approach zero, giving
\begin{equation}
\dfrac{(\psi_n,I\psi_n)}{1-v(\psi_n,w^1\psi_n)}\rightarrow i\gamma \omega'   \, .
\end{equation}
Thus, if we can prove that each number in this sequence is bounded between $a(1{-}v)$ and $b/(1{-}v)$, we are done.

Let $\{\phi_I\}_{I=0}^4$ be an orthonormal basis of $\text{span}(1,p^\nu)$, and define $\Bar{\psi}_n=\psi_n-\sum_I (\phi_I,\psi_n)\phi_I$. Then, we have that
\begin{equation}
\Bar{\psi}_n \in \{1,p^\nu\}^{\perp} \, , \quad\qquad (\Bar{\psi}_n,I\Bar{\psi}_n)=(\psi_n,I\psi_n) \, , \quad\qquad (\Bar{\psi}_n,\Bar{\psi}_n)=1-\sum\nolimits_I |(\phi_I,\psi_n)|^2 \, , 
\end{equation}
and the upper bound is immediately proven:
\begin{equation}
\dfrac{(\psi_n,I\psi_n)}{1-v(\psi_n,w^1\psi_n)} =\dfrac{(\Bar{\psi}_n,I\Bar{\psi}_n)}{1-v(\psi_n,w^1\psi_n)} \leq \dfrac{b (\Bar{\psi}_n,\Bar{\psi}_n)}{1-v} \leq \dfrac{b}{1-v}\, . 
\end{equation}
To obtain the lower bound, we also need the second condition of \eqref{gringooogrigoo}, namely $(\phi_I,\psi_n)=-v(\phi_I,w^1\psi_n)$, so that
\begin{equation}
\begin{split}
&\dfrac{(\psi_n,I\psi_n)}{1-v(\psi_n,w^1\psi_n)}= \dfrac{(\Bar{\psi}_n,I\Bar{\psi}_n)}{1-v(\psi_n,w^1\psi_n)} \geq \dfrac{a(\Bar{\psi}_n,\Bar{\psi}_n)}{1+v}=\dfrac{a}{1+v} \left[1- \sum_I |(\phi_I,\psi_n)|^2\right] \\
& = \dfrac{a}{1+v} \left[1-v^2 \sum_I |(\phi_I,w^1\psi_n)|^2\right] \geq \dfrac{a}{1+v} \left[1-v^2 ||w^1 \psi_n||^2\right] \geq \dfrac{a}{1+v} \left[1-v^2 \right] =a(1-v) \, . \\
\end{split}
\end{equation}

\vspace{-0.2cm}
\subsection{Additional remarks}
\vspace{-0.2cm}

Our proof above can be safely generalized to a broader class of kinetic theories. For example, one can relax the assumption that the particle number is conserved. In that case, the proof is essentially identical: we just need to remove $\delta J^\mu$ from the conserved currents, and $1$ from the kernel of $I$. We can also allow for degenerate statistics, in which case one should replace $f_{\text{eq}} d^3 p$ with $g_s f_{\text{eq}}(1\pm f_{\text{eq}}) d^3p$ in the inner products, where $g_s$ is the spin degeneracy.

\newpage
\section*{Boosted non-hydrodynamic spectrum of randomly-generated models}

\vspace{-0.2cm}
\subsubsection{Without rotational symmetry in the rest frame}
\vspace{-0.2cm}

We generate random models according to the following procedure. Since one may always choose a basis in $\Psi$--space such that $\E^0\,{=}\,1$ and $\sigma\,{=}\,\mathrm{diag}(\sigma_n)$, we impose these conditions from the outset. The non-zero diagonal entries $\sigma_n$ are then sampled as independent random variables uniformly distributed between $0$ and $1$ (in natural units).
To construct $\E^1$, we generate a random matrix $\mathbb{M}$ with entries uniformly distributed in the interval $[-1,1]$, and set $\E^1 \,{=}\, w\,(\mathbb{M}{+}\mathbb{M}^T)/\lVert \mathbb{M}{+}\mathbb{M}^T\rVert$. This normalization ensures that $\lVert \E^1\rVert = w$.

In figure \ref{fig:Randomuzko}, we plot the non-hydrodynamic frequencies $i\omega'$ (evaluated at $k'=0$) as functions of the boost velocity $v$ (with sign), for models with increasing $w$. In all models considered, we take $\Psi$ to be five-dimensional and assume the presence of a single conservation law, resulting in four non-hydrodynamic modes. As is apparent, the spectrum admits a broad range of possible transformation behaviors, since different modes may approach, merge, and cross as $v$ is increased (where the likelihood of close encounters increases with $w$). Moreover, since $\E^1$ is taken to be completely random, the medium is not isotropic in the rest frame, so the spectrum is not invariant under $v\to -v$.

\begin{figure}[b!]
    \centering
\includegraphics[width=0.44\linewidth]{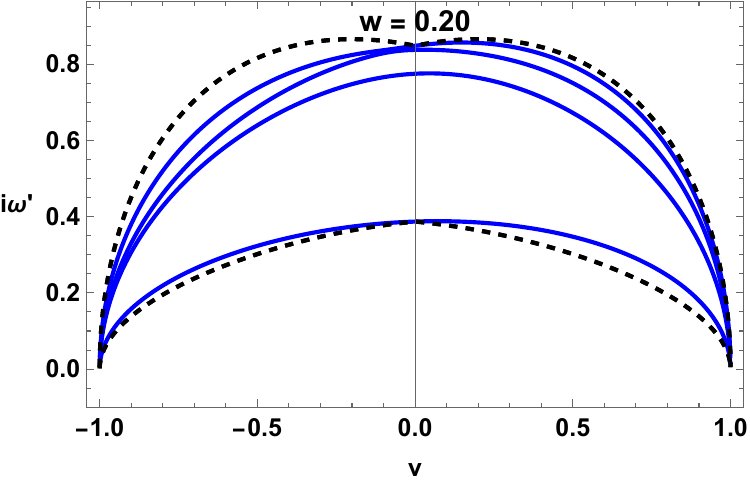}
\includegraphics[width=0.44\linewidth]{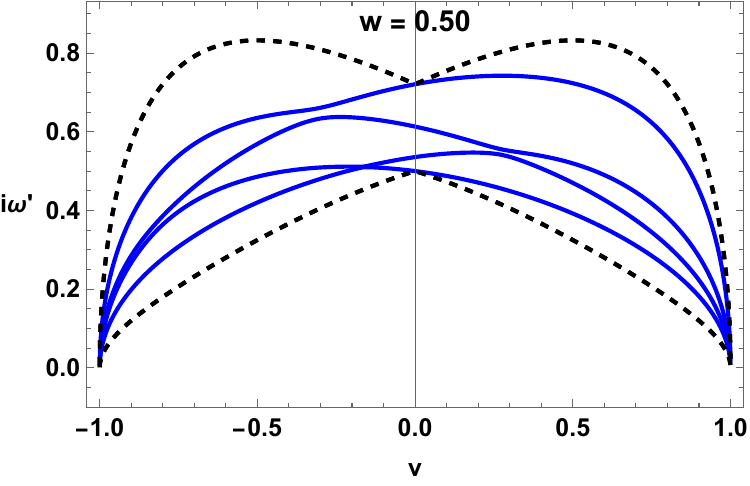}
\includegraphics[width=0.44\linewidth]{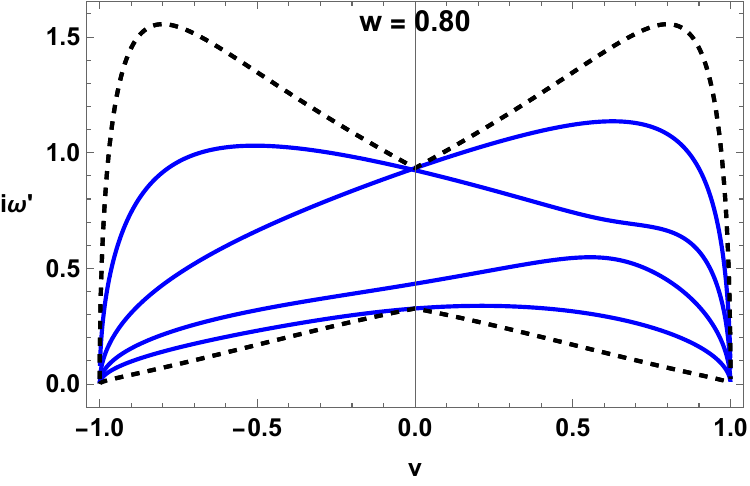}
\includegraphics[width=0.44\linewidth]{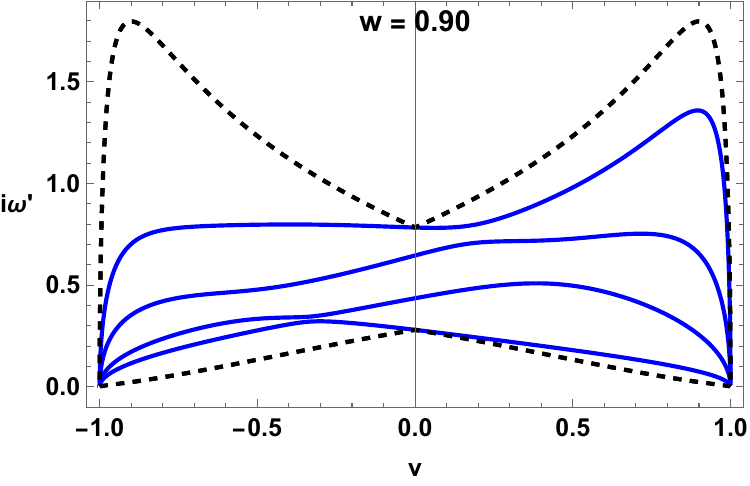}
\includegraphics[width=0.44\linewidth]{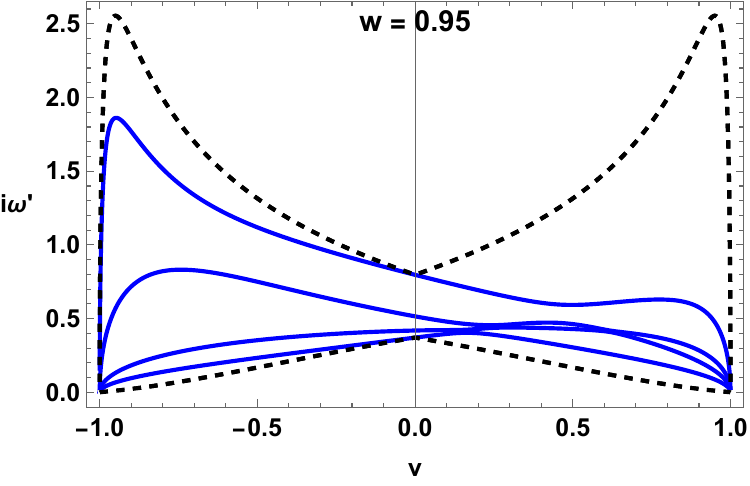}
\includegraphics[width=0.44\linewidth]{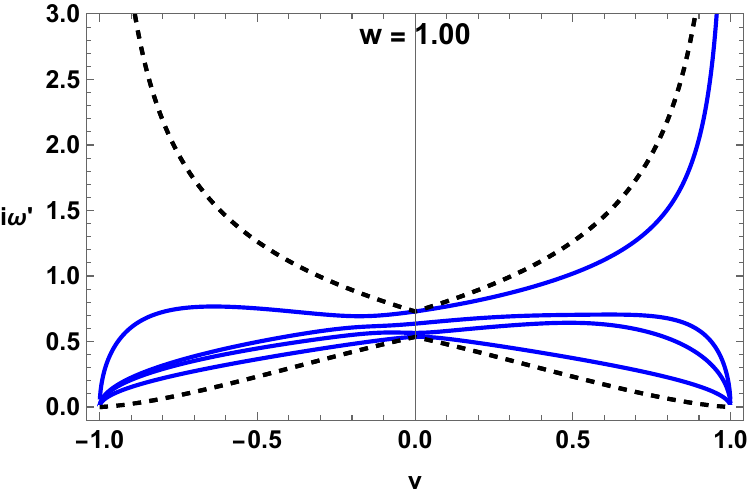}
\caption{Nonhydrodynamic modes at zero wavenumber for a medium moving along the $x$ direction with velocity $v$. The blue curves show the individual relaxation rates $i\omega'$ as functions of $v$, while the dashed curves delimit the region permitted by Theorem~2. Each panel corresponds to a distinct randomly generated model with $D\,{=}\,5$, $N\,{=}\,1$, and the indicated value of $w$.}
    \label{fig:Randomuzko}
\end{figure}

\newpage
\subsubsection{With rotational symmetry in the rest frame}

Assume that the medium is isotropic in its rest frame. Then, if $\Psi(t,x^1)$ is a solution of the equation of motion $(\partial_t+\E^1\partial_1+\s)\Psi=0$, it follows that $\mathbb{L}\Psi(t,-x^1)$ is also a solution, where $\mathbb{L}$ is a $D\times D$ matrix implementing a $180^\circ$ rotation about the $x^3$ axis in $\Psi$--space (with $\mathbb{L}^2=1$). This symmetry implies the matrix relations $\E^1\mathbb{L}=-\mathbb{L}\E^1$ and $\s\mathbb{L}=\mathbb{L}\s$. Recalling that $\s$ is diagonal, and assuming that three variables are $\mathbb{L}$--even  (one of which is conserved) and two are $\mathbb{L}$--odd, the matrices take the form
\begin{equation}\label{isotropizziamo}
\mathbb{L}=
\begin{bmatrix}
1 & 0 & 0 & 0 & 0 \\
0 & 1 & 0 & 0 & 0 \\
0 & 0 & 1 & 0 & 0 \\
0 & 0 & 0 & -1 & 0 \\
0 & 0 & 0 & 0 & -1
\end{bmatrix},
\qquad
\E^1=
\begin{bmatrix}
0 & 0 & 0 & e_{41} & e_{51} \\
0 & 0 & 0 & e_{42} & e_{52} \\
0 & 0 & 0 & e_{43} & e_{53} \\
e_{41} & e_{42} & e_{43} & 0 & 0 \\
e_{51} & e_{52} & e_{53} & 0 & 0
\end{bmatrix},
\qquad
\s=
\begin{bmatrix}
0 & 0 & 0 & 0 & 0 \\
0 & \sigma_1 & 0 & 0 & 0 \\
0 & 0 & \sigma_2 & 0 & 0 \\
0 & 0 & 0 & \sigma_3 & 0 \\
0 & 0 & 0 & 0 & \sigma_4
\end{bmatrix},
\end{equation}
where $e_{mn}$ and $\sigma_n$ are random variables. The resulting boosted spectra are shown in figure \ref{fig:Randomuzko2}. We see that, in this case, the spectra are perfectly symmetric under the transformation $v\to -v$, as expected when an isotropic medium is set into uniform motion.

In the upper panels, all coefficients $\sigma_n$ are chosen as \textit{independent} random variables. Consequently, there are four distinct nonhydrodynamic frequencies at $v=0$. Since the frequencies $\omega'$ depend smoothly on $v$ and the spectrum is even under $v\to -v$, all blue curves necessarily have vanishing slope at $v=0$. As a result, these models cannot approach the bounds tangentially at small $v$, because the bounds themselves have finite slope at the origin.

In the lower panels, we instead impose the constraints $\sigma_1=\sigma_3$ and $\sigma_2=\sigma_4$, so that the nonhydrodynamic spectrum consists of two modes that are twofold degenerate at $v=0$. When $v$ is turned on, each degenerate pair can split into branches with opposite slopes, thereby preserving the even symmetry of the spectrum. This allows the modes to adhere to the bounds in the small-$v$ regime.

\begin{figure}[b!]
    \centering
\includegraphics[width=0.44\linewidth]{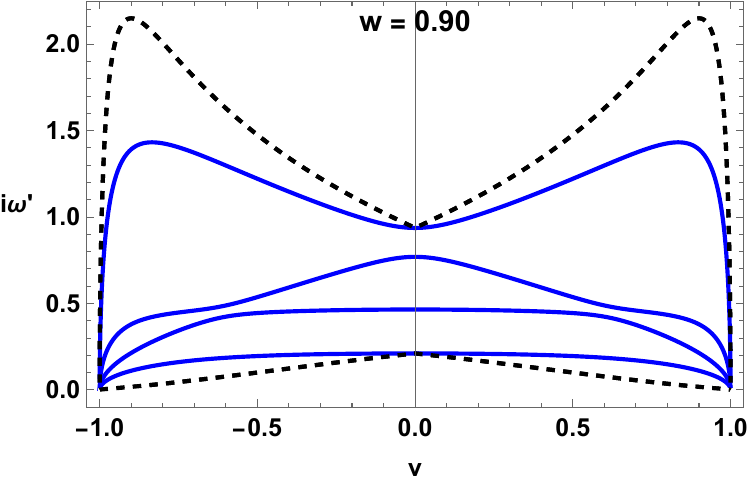}
\includegraphics[width=0.44\linewidth]{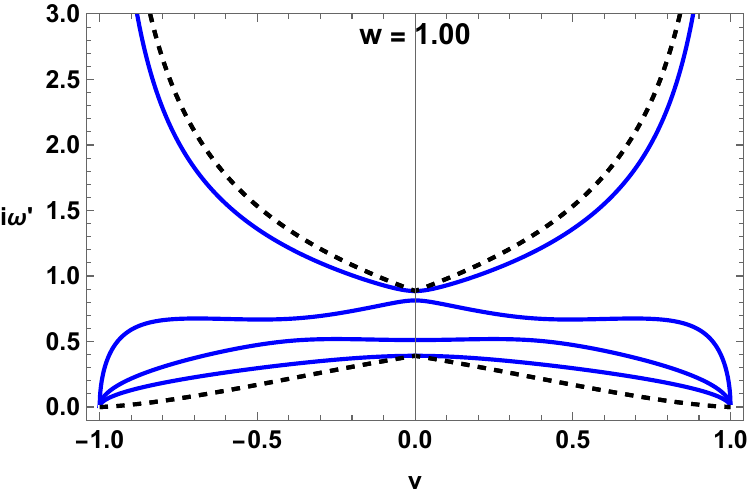}
\includegraphics[width=0.44\linewidth]{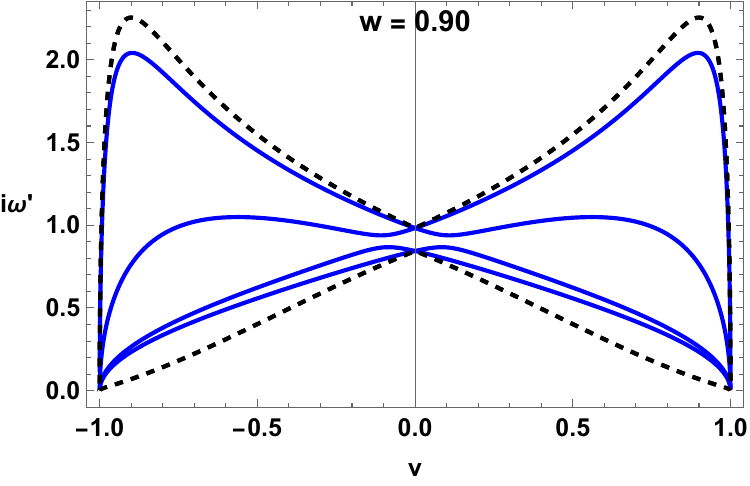}
\includegraphics[width=0.44\linewidth]{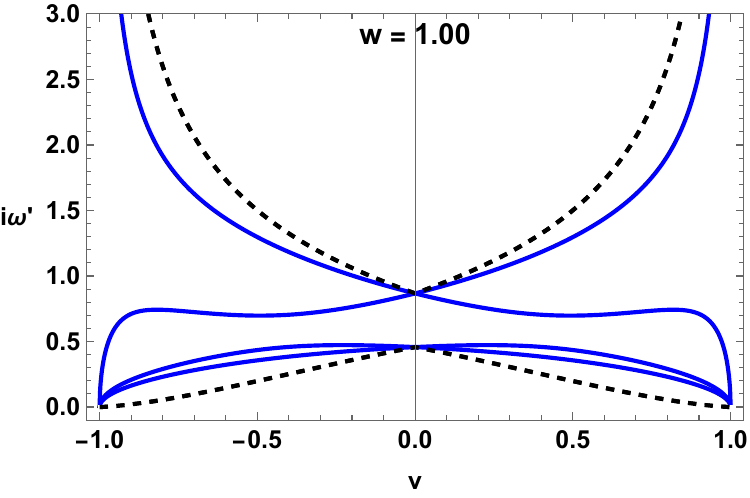}
\caption{Nonhydrodynamic modes at zero wavenumber for a medium moving along the $x$ direction with velocity $v$. The blue curves display the individual relaxation rates $i\omega'$ as functions of $v$, while the dashed curves indicate the region allowed by Theorem~2. Each panel corresponds to a distinct randomly generated model of the form \eqref{isotropizziamo}. In the upper panel, the coefficients $\sigma_n$ are chosen as completely independent random variables. In the lower panel, the $\sigma_n$ are random variables subject to the constraints $\sigma_1=\sigma_3$ and $\sigma_2=\sigma_4$.}
    \label{fig:Randomuzko2}
\end{figure}

\section*{Breakdown of quasi-hydrodynamics at large boosts}

We consider a rest-frame model of the form $(\partial_t+\E^1\partial_x+\s)\Psi=0$, with $\Psi\in\mathbb{C}^D$, whose maximal signal velocity is, say, $w\,{=}\,1/2$, and whose rest-frame non-hydrodynamic spectrum exhibits a gap $a=1$. According to Theorem~2, after a boost with velocity $v$, the non-hydrodynamic frequencies retain a minimal gap $a'=(1-|v|/2)/\gamma$.

We now introduce an extended model $\big[\partial_t+\E^1_{(+2)}\partial_x+\s_{(+2)}\big]\Psi_{(+2)}=0$, with $\Psi_{(+2)}\in\mathbb{C}^{D+2}$, defined by
\begin{equation}\label{quasuzzo}
\E^1_{(+2)} = \left[\begin{array}{c|cc}
\E^1 & 0 & 0 \\
\hline
0 & 1 & 0 \\
0 & 0 & -1
\end{array}\right]\,, 
\qquad\qquad 
\s_{(+2)} = \left[\begin{array}{c|cc}
\s & 0 & 0 \\
\hline
0 & \lambda & 0 \\
0 & 0 & \lambda
\end{array}\right]\qquad\qquad (\text{with }0<\lambda\ll 1)\, .
\end{equation}
This extension amounts to adding two non-hydrodynamic degrees of freedom that propagate at the speed of light and decay parametrically slowly, with rate $\lambda$, in the rest frame.
In a frame moving with velocity $v$, the non-hydrodynamic spectrum of this extended model coincides with that of the original system, supplemented by two additional modes,
\begin{equation}
i\omega'_{\pm}=\dfrac{\lambda}{\gamma(1\pm v)} \, .
\end{equation}
These modes lie in the quasi-hydrodynamic regime as long as they remain well separated from the rest of the spectrum. This separation breaks down when $|v|\gtrsim 1-2\lambda$. Therefore, regardless of how small $\lambda$ is, there always exists an observer $\mathcal{O}'$ for whom no quasi-hydrodynamic spectral separation is present at $k'=0$. In the limit $\lambda\to0$, such an observer must move increasingly close to the speed of light.

\begin{figure}[h!]
    \centering
\includegraphics[width=0.5\linewidth]{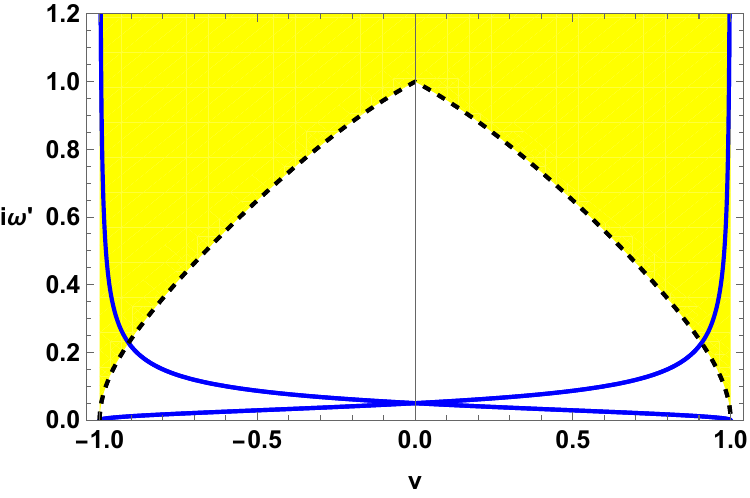}
\caption{Non-hydrodynamic spectrum of model \eqref{quasuzzo}. A large population of modes has minimal relaxation rate $1$, and maximal propagation speed $1/2$. At finite boosts, these modes populate the yellow region, bounded from below by the curve $(1-|v|/2)/\gamma$ (dashed). In addition, there are two fully decoupled quasi-hydrodynamic modes, with rest-frame relaxation rates $i\omega_{\pm}=\lambda\ll 1$, and which propagate exactly at the speed of light. At sufficiently high boosts, one of these modes experiences an arbitrarily large blueshift, and thus decays faster than the rest of the spectrum.}
    \label{fig:quasiHydro}
\end{figure}

\label{lastpage}
\end{document}